\documentclass[prx,twocolumn,aps,epsf,showpacs,superscriptaddress,longbibliography,nofootinbib,notitlepage]{revtex4-2}
\usepackage{amsmath,mathtools,amssymb,mathrsfs}
\usepackage[pdftex]{graphicx}
\usepackage[dvipsnames]{xcolor}
\usepackage[colorlinks=true,allcolors=blue]{hyperref}
\usepackage{float}\usepackage{xcolor}
\usepackage{bm}        % for math
\usepackage{braket}
\usepackage{epstopdf}
\usepackage{xfrac}
\usepackage{cancel}
\usepackage{soul}
\usepackage{amsthm}
\usepackage{physics}
\usepackage{amsfonts}
\usepackage{bbm}
\usepackage{enumerate}
\usepackage{setspace}
\usepackage{url}  % This makes \url work
\usepackage{mathrsfs}
\usepackage[normalem]{ulem}
\usepackage{accents}

\newcommand{\<}{\langle}

\newcommand{\up}{\uparrow}
\newcommand{\down}{\downarrow}

\renewcommand{\>}{\rangle}
\renewcommand{\(}{\left(}
\renewcommand{\)}{\right)}
\renewcommand{\[}{\left[}
\renewcommand{\]}{\right]}
\renewcommand{\d}{\partial}

\newcommand{\eps}{\epsilon}

\usepackage{titlesec}
\titleformat{\paragraph}[runin]% runin puts it in the same paragraph
        {\bfseries}% formatting commands to apply to the whole heading
        {}%{\thesection.\thesubsection.\thesubsubsection)}%{\thesection.\thesubsection\thesubsubsection) }% the label and number
        {0pt}% space between label/number and subsection title
        {\noindent}% formatting commands applied just to subsection title
        [~--]% punctuation or other commands following subsection title
\titlespacing*{\paragraph}{0pt}{6pt}{4pt}

\newcommand{\acp}[1]{\textcolor{Aquamarine}{#1}}
\newcommand{\yz}[1]{\textcolor{cyan}{#1}}

\begin{document}
\title{Sequential quantum simulation of spin chains with a single circuit QED device}
%\title{Quantum simulation of spins on a single circuit QED device
%\acp{We should re-think this, it currently sounds like an experimental work, also perhaps we could mention ``control" }}

\author{Yuxuan Zhang}
\affiliation{Department of Physics, University of Texas at Austin, Austin, TX 78712, USA}
\affiliation{Department of Physics and Centre for Quantum Information and Quantum Control, University of Toronto,
60 Saint George St., Toronto, ON M5S 1A7, Canada}

\author{Shahin Jahanbani} 
\affiliation{Department of Physics, University of Texas at Austin, Austin, TX 78712, USA}
\affiliation{Department of Physics, University of California, Berkeley, California 94720, USA}
\author{Ameya Riswadkar}
\author{S. Shankar}
\affiliation{Department of Electrical and Computer Engineering, University of Texas at Austin, Austin, TX 78712, USA}

\author{Andrew C. Potter}
\affiliation{Department of Physics and Astronomy, and Stewart Blusson Quantum Matter Institute, University of British Columbia, Vancouver, BC, Canada V6T 1Z1}

\begin{abstract}
Quantum simulation of many-body systems in materials science and chemistry are promising application areas for quantum computers. However, the limited scale and coherence of near-term quantum processors pose a significant obstacle to realizing this potential. 
Here, we theoretically outline how a single-circuit quantum electrodynamics (cQED) device, consisting of a transmon qubit coupled to a long-lived cavity mode, can be used to simulate the ground state of a highly-entangled quantum many-body spin chain. We exploit recently developed methods for implementing quantum operations to sequentially build up a matrix product state (MPS) representation of a many-body state.
This approach re-uses the transmon qubit to read out the state of each spin in the chain and exploits the large state space of the cavity as a quantum memory encoding inter-site correlations and entanglement.
We show, through simulation, that analog (pulse-level) control schemes can accurately prepare a known MPS representation of a quantum critical spin chain in significantly less time than digital (gate-based) methods, thereby reducing the exposure to decoherence.
We then explore this analog-control approach for the variational preparation of an unknown ground state. We demonstrate that the large state space of the cavity can be used to replace multiple qubits in a qubit-only architecture, and could therefore simplify the design of quantum processors for materials simulation.
We explore the practical limitations of realistic noise and decoherence and discuss avenues for scaling this approach to more complex problems that challenge classical computational methods.
\end{abstract}
\maketitle

Achieving control over quantum state spaces that are too large to explore classically is a central challenge of quantum computing. The most common approach has been to build quantum devices out of arrays of two-level qubits. This task requires achieving coherent control of $\gtrsim 40$ interacting qubits to exceed exact classical simulation, and presents significant challenges for fabrication, tune-up, and cross-talk minimization~\cite{aaronson2011computational,aaronson2016complexity,arute2019quantum}. By contrast, utilizing more than two quantum states per device can reduce the size of processors required to exceed classical simulation. Circuit quantum electrodynamics (cQED) devices~\cite{blais2021circuit,vandersypen2005nmr} consists of a transmon qubit interacting with a superconducting cavity that can be modeled as a many-level quantum harmonic oscillator. In these devices, the non-linear coupling between the qubit and cavity enables universal control over the qubit and cavity. Such universal control has been employed for the preparation of non-Gaussian states of the oscillator~\cite{kudra2022robust}, enabling, for example, implementation of a quantum error correcting code on a single device~\cite{gottesman2001encoding,mirrahimi2014dynamically,nigg2014deterministic,heeres2017implementing}. 
Few-mode cQED devices have also been exploited for proof-of-concept few-body chemistry simulations~\cite{wang2020efficient,wang2023observation}. These demonstrations raise the question: Can the larger state space of cQED devices be leveraged for even larger scale simulations of complex quantum many-body problems relevant to condensed matter physics and material science?

The reduction of hardware complexity for performing classically inaccessible simulations can be further enhanced by hardware-efficient approaches to quantum algorithms. Quantum circuit tensor network state (qTNS) techniques~\cite{schon2005sequential,barratt2021parallel,foss2021holographic}, use repeated mid-circuit qubit reset and reuse to sequentially simulate many-body quantum states. 
By exploiting the efficient compression~\cite{orus2019tensor} of physically interesting states, such as low-energy states of local Hamiltonians, qTNS methods enable simulation of many-body systems relevant to condensed-matter physics and materials science with much smaller quantum memory than would be required to directly encode the many-body wave-function.
Rather than directly encoding the quantum many-body wave function into qubits, sequential simulation involves implementing a sequence of quantum operations that allows one to sample properties of the many-body state along a spatial direction, without ever storing the full state in quantum memory.

In this work, we theoretically explore the synthesis of these two approaches to reduce the hardware requirements to perform complex materials simulations. Specifically, we simulate the use of a single cQED device consisting of a transmon qubit coupled to a long-lived cavity for variational preparation of 1d quantum circuit matrix product states (qMPS) that approximately represent the highly-entangled ground state of a critical, non-integrable spin chain. In this approach, the transmon qubit is used to represent the state of a spin on a single site and is repeatedly reset and reused. The cavity operates as a many-level quantum memory that retains coherent information about inter-site correlations and entanglement. 

\begin{figure*}
    \centering.
    \includegraphics[width=0.95\textwidth]{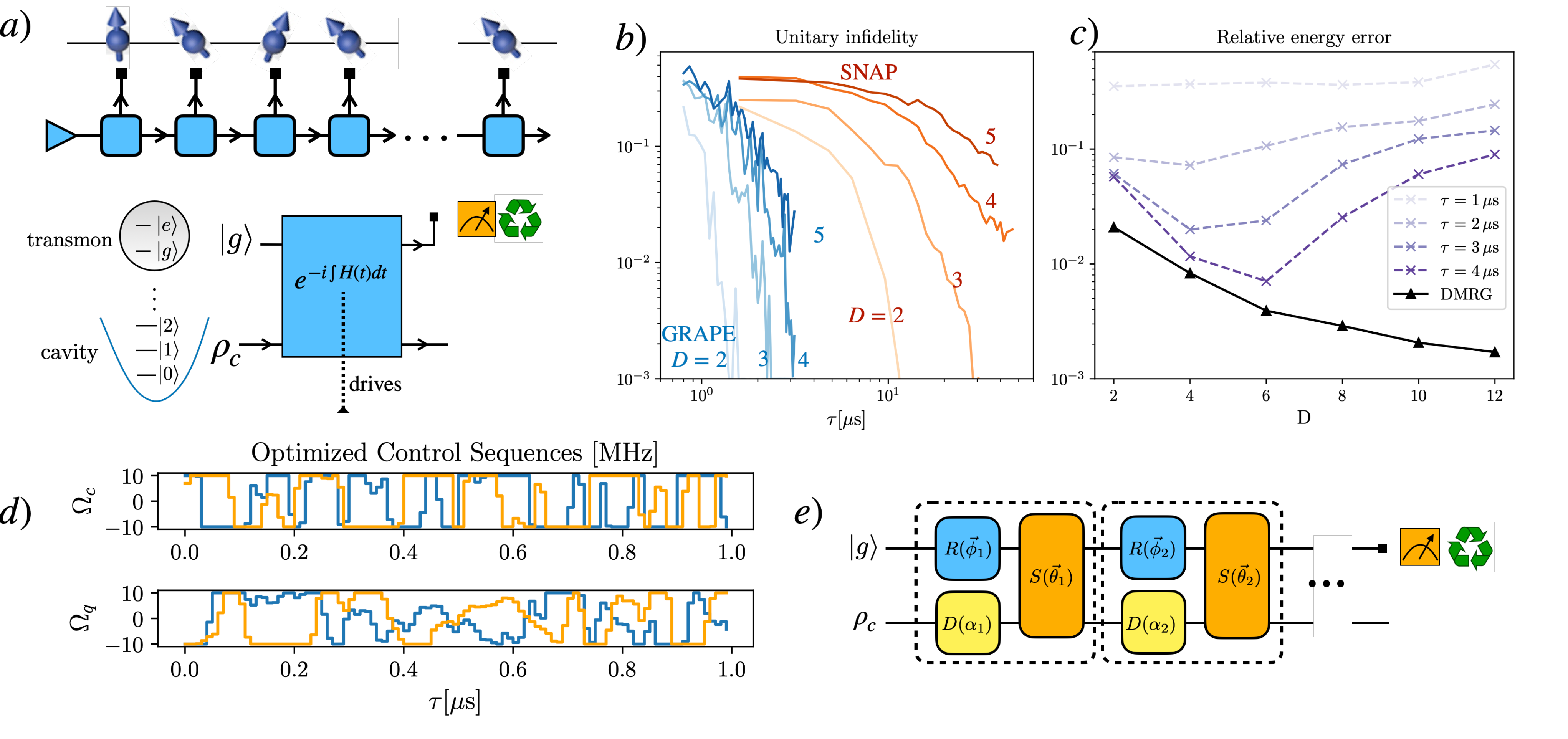}
    \vspace{-12pt}
    \caption{{\bf Sequential preparation of a qMPS on a cQED device} a) Graphical representation of a matrix product state for a spin chain (top) and its implementation as a sequential circuit with qubit reuse. Boxes represent three-index tensors with physical index (vertical line ending in box) representing the state of a spin-1/2 physically implemented by a transmon qubit, and bond indices (horizontal lines) representing inter-site entanglement, physically implemented by a superconducting cavity. Each tensor is implemented with a unitary time evolution entangling the qubit and cavity, followed by measuring the transmon qubit and resetting to its ground state $|g\>$, which is then reused for the next site in the chain. Arrows indicating causality of circuit implementation are opposite to the conventional isometry arrows.  b) Infidelity for synthesizing an isometry representing various bond-dimension $D$ approximations of a critical Ising chain ground state via GRAPE pulse-level control (d) or discrete circuits based on SNAP gates (e). The GRAPE-based method requires significantly less time to achieve the same accuracy, making it less susceptible to decoherence. c) Relative (fractional) energy error for the resulting ground state of Eq.~\ref{eq:sdim} for the GRAPE-based method for various bond dimensions and time length of the waveform.  }
    \label{fig:schematic}
\end{figure*}

A key challenge is to control the interactions between the qubit and cavity to implement unitary operations that represent the tensors of the MPS. Previous works largely focused on the use of the transmon as an ancillary qubit to achieve non-linear control over the cavity without storing information. By contrast, for this application, we will need to coherently control both the qubit and oscillator in tandem. 
To this end, we explore two approaches to controlling the cQED device to implement variational qMPS simulations: i) ``analog" pulse-level control in which the qubit and cavity drive waveforms are treated as variational parameters~\cite{heeres2017implementing}, and ii) a ``digital" gate-based approach in which control pulses are pre-compiled into a discrete set of (parameterized) gates that are concatenated into a circuit~\cite{heeres2015cavity,krastanov2015universal,fosel2020efficient,lin2022efficient}.

Our work shows that, despite the much larger number of variational parameters that must be optimized, the pulse-level approach dramatically outperforms the gate-based approach in the time (and hence error rate) needed to synthesize unitaries relevant to prepare a physical qMPS. Specifically, the analog approach enables high-fidelity synthesis in nearly an order of magnitude shorter time than the gate-based method, thereby significantly reducing the exposure to noise and decoherence.

We then explore the complexity of using the analog control approach to variationally optimize qMPS approximations of a highly entangled critical spin-chain ground state, including assessing the effects of realistic levels of decoherence. 
This approach synthesizes the gate-free ctrl-VQE method of~\cite{meitei2021gate}, with qubit-efficient qMPS methods.
We show that the number of cavity levels that can be effectively used is ultimately limited by decoherence, and estimate that with current technology, a single cQED device can implement qMPS representations with bond dimensions up to $6$, which would require $4$ qubits in a qubit-only architecture.
Finally, we conclude by discussing pathways for scaling this approach to multi-mode cavities in order to achieve a quantum advantage over classical methods for materials simulation tasks.

\paragraph{Sequential/holographic simulation with a cQED device}
The idea of using sequential circuits to simulate 1d many-body systems was first introduced in~\cite{schon2005sequential}, used to build simple many-body states in a cavity QED setup~\cite{schon2005sequential}, and later generalized into a framework for variational ground state preparation~\cite{kim2017holographic,kim2017noise,kim2017robust,foss2021holographic}, quantum dynamics simulation~\cite{foss2021holographic,lin2021real,astrakhantsev2022time}, and higher-dimensional qTNS~\cite{zaletel2020isometric,lin2022efficient}. Here, we briefly review the key ideas of this method by describing its implementation for simulating a 1d spin-1/2 chain with a single cQED device. We restrict our attention to the two lowest energy qubit states of the transmon: $|0\>,|1\>$, and define the corresponding transmon Pauli operators $\vec{\sigma}$. We model the cavity as a quantum harmonic oscillator with annihilation operators $a$ and occupation number $n=a^\dagger a$.

As shown in Fig.~\ref{fig:schematic}a, rather than preparing the wave function of $L$-spins on $L$ independent qubits, in sequential simulation, one instead re-uses a single qubit (here, the transmon) for each physical spin and utilizes a small quantum memory (here, the cavity mode) to coherently store information about correlations and entanglement between spins. 
The simulation proceeds by initializing the transmon qubit into a fixed reference state, $|0\>$, entangling it with the cavity through a unitary $U$. The qubit then holds the state of the first spin in the chain, and it can be measured in any desired basis. 
The qubit is then reset to $|0\>$ and the process is iterated to prepare the second spin in the chain, and so on. 

In this way, one sequentially prepares the state of the spin chain from left to right and may sample any desired correlation function along the way.
Formally, this procedure is equivalent to sampling correlations from a resulting spin-chain state $|\psi\>$ with MPS representation:
\begin{align}
|\psi\> = \sum_{s_1\dots s_L} A^{s_1}A^{s_2}\dots A^{s_L} |s_1\dots s_L\>
\end{align}
with tensors: $A^s_{i,j} = \<s,j|U|0,i\>$ where $s=0,1$ correspond to spin $\up,\down$ respectively, and $i,j=0, 1, 2\dots$ index the states of the cavity, which can be viewed as a physical representation of the virtual bond-space of the MPS (see~\cite{foss2021holographic} for a detailed discussion of boundary conditions). 
Throughout, for convenience, we focus on translation-invariant qMPS with the same $A,U$ for each site (though this is not essential for the method).
Though the oscillator technically possesses an infinite state-space, in practice, the large-$n$ states decohere more rapidly ultimately limiting the number of usable levels to a finite number,  $D$.
In the qMPS context, $D$ corresponds to the bond dimension of the MPS, $|\psi\>$.

For example, to measure the correlation function: $\<\psi|\sigma^z_r\sigma^x_1|\psi\>$, where $|\psi\>$ is the state of the spin chain, one measures the qubit in the $x$ basis after the first implementation of $U$, then iteratively applies $U$ and resets the qubit $r-1$ times, and finally measure the qubit in the $z$ basis, which gives one statistical sample. The average of many such samples gives the correlation function. Note that, there is no additional overhead to sampling correlations in this manner compared to the case where one directly encodes $|\psi\>$ onto $L$ qubits.

\begin{figure*}
    \centering.
    \includegraphics[width=1\textwidth]{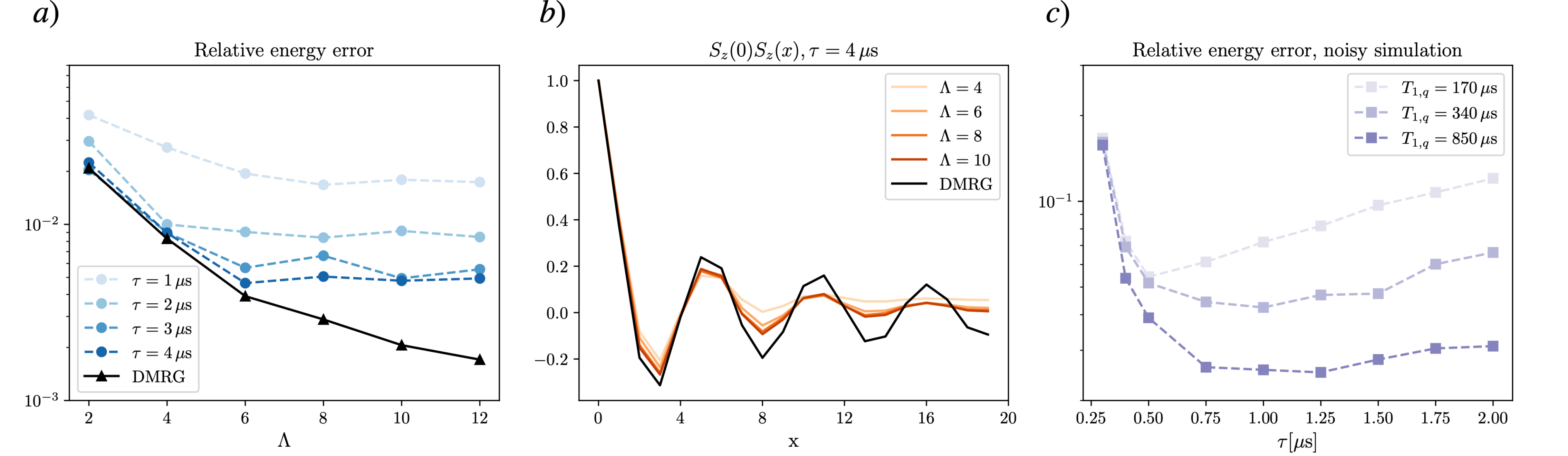}
    \caption{{\bf VQE numerical optimization results} a) Ideal VQE optimization energy result at different $\tau$, compared with classical density matrix renormalization group (DMRG) simulation. The x-axis stands for the artificial cutoff that we choose for the cavity in our simulation. b) Two-point correlation function $S_z(0)S_z(x)$ for the optimized VQE results for $\tau = 4\ \mu s$ at different finite cutoff levels. c) Noisy VQE optimization energy result versus $\tau$ in three different hypothetical transmon relaxation times. The relaxation times are chosen to be an integer multiplier of a reported hardware parameter~\cite{heeres2017implementing}: $T_{1,c} = 2700\ \mu s $, $T_{1,q} = 170\ \mu s$ and $T_{2,q} = 43\ \mu s$.% \acp{Give numerical value for ratio of $T_{1,2}^{c/q}$ to $T_1$, and also should label $T_1$ by a $c/q$ label for clarity}
    }
    \label{fig:result}
\end{figure*}

Though it is possible to classically simulate 1d ground states using MPS methods, there are significant computational challenges for computing non-equilibrium dynamical properties and properties of highly entangled ground states in higher dimensions. Here, qTNS methods may play a role -- allowing quantum computers to benefit from the efficient compression of qTNS, while potentially exceeding classical TNS methods. 
In this vein, we remark in passing that the sequential circuit concept can be generalized to implement various tensor network geometries and higher dimensions: by using a $d$-dimensional qubit array, and an auxiliary quantum memory, one can sequentially prepare a $d+1$ dimensional tensor network state with bond-dimension equal to the Hilbert space size of the quantum memory. 
Due to the dimensional mismatch between the quantum processor and the state being simulated, this approach is sometimes referred to as ``holographic" simulation (not to be confused with other contexts of holography in physics~\cite{susskind1995world}). 
A limitation is that only isometric tensor networks~\cite{soejima2020isometric} can be implemented using physically allowed quantum operations. In 1d, the isometry constraint is known not to be a serious limitation. Its role in higher dimensions is the subject of active ongoing investigation~\cite{zaletel2020isometric,slattery2021quantum}.

\paragraph{cQED control approaches}
A key challenge in holographic simulation is to find an implementation of $U$ that represents the tensors, $A^s_{ij}$, corresponding to an MPS representation of states of physical interest. For the cQED device, $U = \mathcal{T}e^{-i\int_0^\tau H_{\rm cQED}(t)dt}$ is generated by time evolution under the cQED Hamiltonian, $H_{\rm cQED} = H_0+H_1+H_{\rm drive}$:
\begin{align}
H_0 &= \frac12 \omega_q\sigma^z+\omega_c n
\nonumber\\
H_1 &= \frac{K}{2} (a^\dagger)^2a^2+ \frac{\chi}{2} n\sigma^z+\frac{\chi'}{2}(a^\dagger)^2a^2\sigma^z
\nonumber\\
H_{\rm drive}(t) &= \Omega_c a + \Omega_q\sigma^- +{\rm h.c.}
\end{align}
where $\vec{\sigma}$ are Pauli matrices for the qubit, $a,n=a^\dagger a$ are respectively annihilation and number operators of the cavity. 
%
% $H_{\rm cQED}(t) = H_0+V + H_{\rm drive}(t)$. In the rotating frame of the qubit and cavity drives (setting $\hbar=1$ throughout), these read %
%\begin{align}
%H_{\text{0}} =& \omega_q \hat{b}^\dagger\hat{b}^{\vphantom \dagger}+ \omega_c a^\dagger a^{\vphantom \dagger}
%\nonumber\\
%V=&\frac{K}{2}(a^{\dag})^2(a)^2 + \frac{\alpha}{2}(\hat{b}^{\dag})^2(\hat{b})^2+ \chi a^{\dag}a \hat{b}^{\dag}\hat{b}+
%\nonumber\\&
% + \frac{\chi'}{2}(a^{\dag})^2(a)^2 \hat{b}^{\dag}\hat{b}
%\nonumber\\
%H_{\text{drive}}(t) =& \Omega_{c}(t)a +\Omega_{q}(t)\hat{b} + \text{h.c.}
%\end{align}
Simulations are performed in the rotating frame of $H_0$, leaving only $H_1+H_{\rm drive}(t)$.
For realism, we adopt the parameters from~\cite{heeres2017implementing}, and choose dispersive shift $\chi=-2\pi\times2194$ kHz, second order dispersive shift $\chi'=-2\pi\times19$ kHz and cavity non-linearity $\alpha=-2\pi\times3.7\,$ kHz. These parameters are fixed by the device design~\cite{heeres2017implementing}. 

Control is achieved through the (complex) drive waveforms for the qubit (q) and cavity (c): $\Omega_{q,c}(t)$. 
Due to the practical limitations of the wave-form generators and to deal with finite-dimensional pulse optimizations, we focus on piecewise constant drives with $N_{\rm ts} = \tau/\Delta t$ time-steps of size $\Delta t$, such that:
\begin{align}
U\[\Omega_{c,q}(t)\] = \prod_{j=1}^{N_{\rm ts}} e^{-iH_{\rm cQED}(t_j)\Delta t}.
\end{align}
We choose $\Delta t = 10\ {\rm ns}$, which is an order of magnitude higher than the typical minimum pulse control resolution, to reduce the number of variational parameters and smoothen the waveform.
To model realistic limitations on drive range, we restrict the drive amplitude to $|\Omega_{c,q}|\leq \Omega_{\rm max}=10$ MHz.

We compare two strategies to implement a desired $U$. First, we consider an ``analog"  scheme, similar to gradient ascent pulse engineering (GRAPE)~\cite{law1996arbitrary,khaneja2005optimal}, in which the individual (discretized) drive amplitudes, $\Omega_{c,q}(t_j)$ are treated as variational parameters. This method allows for potentially highly efficient control, at the expense of introducing a large number, $4N_{\rm ts}$, of variational parameters for the real and imaginary parts of $\{\Omega_{c,q}(t_j)\}|_{j=1}^{N_{\rm ts}}$. 

Second, we consider pre-compiling the pulse waveforms into a sequence of discrete, parameterized gates. Specifically, we consider alternating layers of qubit rotations $R(\vec{\phi}) = e^{-i\vec{\phi}\cdot \vec{\sigma}}$, cavity displacements $D(\alpha) = e^{\alpha a-\alpha^* a^\dagger}$, and selective number-dependent arbitrary phase (SNAP) gates~\cite{krastanov2015universal,heeres2015cavity,fosel2020efficient,kudra2022robust}: $S(\theta) = e^{i\theta(n)\sigma^z}$ where $\theta(n)$ applies a distinct, arbitrary phase for each occupation number $n$~\cite{heeres2015cavity} of the D-level cavity. The unitary is then composed of layers of these discrete gates: $U = \prod_{i=1}^N D(\alpha_i) R(\vec{\phi_i}) S(\vec{\theta_i})$, where $\{\alpha_i, \vec{\phi}_i, \vec{\theta}_i\}$ are variational parameters. As we demonstrate in Appendix~\ref{supp:snap}, sufficiently deep circuits of this form enable arbitrary control over the joint state of the qubit and oscillator. 

While both the GRAPE and SNAP-based methods have been separately well-studied and used extensively to construct non-Gaussian states of the cavity mode~\cite{hofheinz2009synthesizing,krastanov2015universal,heeres2017implementing,kudra2022robust}, a direct comparison of their performance is lacking. Moreover, many previous studies use the qubit as a sacrificial ancilla that starts and ends in a fixed reference state~\cite{heeres2015cavity,heeres2017implementing,kudra2022robust}, whereas for qMPS applications we need to achieve control over the joint entangled state of the qubit and cavity to realize unitary operations relevant to MPS-tensors for physical ground states. Next, we compare the performance of these methods for sequentially simulating a physical system.

\paragraph{Case study: non-integrable Ising model}
We will focus on different variational approaches for approximating the quantum-critical ground state of a non-integrable Ising model with self-dual perturbation (SDIM)~\cite{rahmani2015phase}:
\begin{align}    
\label{eq:sdim}
H_{\text{SDIM}} &= -\sum_{i} [J\sigma^z_i \sigma^z_{i+1} + h\sigma^x_{i}
    -V  \left(\sigma^x_i \sigma^x_{i+1} + \sigma^z_{i-1} \sigma^z_{i+1} \right) ],
\end{align}
The inclusion of $V$, spoils the exact solvability of the model with $V=0$.
We focus on $J=h=1$ and $V=0.5$ where the ground state is poised at a critical point between magnetically ordered and disordered states, and exhibits power-law decaying correlations and entanglement that diverges with system size. Such critical states can only be approximately captured by any finite bond-dimension MPS (allowing continued room for improvement with increasing bond-dimension, $D$, i.e. with the utilization of more cavity levels).
We consider two different variational tasks: unitary synthesis and variational ground state preparation (described below). 

\paragraph{Unitary synthesis}
In unitary synthesis, we start with a known classical MPS representation of the ground state with tensors $A$ and bond-dimension $D$, block encode $A^{s}_{ij}$ into a unitary, $U_{\rm target}$ using the Gram-Schmidt algorithm (see Appendix~{\ref{supp:opt}} for details), and then maximize the trace-fidelity: 
\begin{align}
\mathcal{F}(\vec{\varphi}) = \left\vert{\rm tr} \(U^\dagger_{\rm target}U(\vec{\varphi})\)\right\vert/{\rm tr}\mathbbm{1}
\end{align}
over the variational parameters, $\vec{\varphi}$. 
This procedure requires a sufficiently low bond dimension to perform each step classically, yet for small instances, it may still serve as a benchmark to compare the two variational approaches. Moreover, there may be contexts where this method can assist in achieving a quantum advantage, for example by seeding a more general variational approach with a known classical approximation in order to simplify a complex variational optimization~\cite{niu2021holographic}. Moreover, there are many settings where ground states can be efficiently prepared classically, but where simulating non-equilibrium dynamical properties (electrical or thermal conductivity, optical absorption spectra, etc.) of the system starting from the ground state may be challenging. Here, efficiently preparing the ground state on a quantum device is an important prerequisite to using quantum algorithms for computing these more challenging quantities. A prerequisite of this approach is that standard results from quantum optimal control theory~\cite{khaneja2005optimal} guarantee that an optimal solution may be found by simple gradient ascent methods, without encountering local minima trapping that can plague other high-dimensional optimization problems. 

%\yz{this coincides with the displacement gate - do we want to use small d?}\acp{I think D is ok b/c displacements always have an argument $D(\alpha)$, we could use $D_b$ if you're worried about it, $d$ usually stands for physical dimension} 

Here, we construct the target unitaries from the MPS tensors with bond-dimension $D$ obtained by standard DMRG algorithm~\cite{white1992density}. Larger $D$ captures a more accurate approximation to the ground state but presents a more challenging unitary synthesis problem, requiring a longer control sequence time, $\tau$.
To directly compare analog and digital variational approaches a key metric will be the time it takes to achieve a given accuracy for the target state since $U$ must be implemented in time $\tau$ that is less than the qubit and cavity coherence times in order to avoid significant decoherence. Following~\cite{kudra2022robust}, we estimate that optimal control methods for synthesizing SNAP and displacement gates require $\sim 800\ {\rm ns}$ per circuit layer (also see Appendix~\ref{supp:snap}.)
%\acp{do we need to justify this, maybe in an appendix?}.

Fig.~\ref{fig:schematic} shows that while infidelity, $1-\mathcal{F}$, decays exponentially with time for both the analog and digital control, the former method converges over an order-of-magnitude faster timescale. 
This apparently apparently reflects overall inefficiency in packaging drive amplitudes into a fixed sequence of SNAP gates, compared to directly optimizing the waveforms to maximize the target fidelity~\cite{meitei2021gate}.
In fact, assuming transmon coherence times of $T_2\approx40\ {\rm \mu}s$~\cite{heeres2017implementing}, it is evident that the gate-based approach will be insufficient to prepare even small bond-dimension states with $D\gtrsim 2$, whereas in this time, the analog control can successfully prepare up to $D\approx 5$ with percent-level infidelity. 
For this reason, in the remainder of the paper, we will focus exclusively on the analog approach. We note, in passing, however, that the large number of variational parameters in GRAPE could lead to longer wall-clock time to optimize -- which will depend in a detailed fashion on the optimization scheme chosen. Here, we simply focus on the time, $\tau$, needed to coherently execute the circuit -- as this is the ultimate limitation of using quantum coherent computations.

While this unitary-synthesis study clearly highlights the advantages of using analog-style pulse-level control over gate-based methods for qMPS applications, in many practical settings one does not have a classical representation of the ground state to work with, and an alternative approach is required.

\paragraph{Variational ground state preparation}
The above unitary synthesis procedure requires a known classical MPS representation of the system. In many practical instances, one instead has only a model Hamiltonian and would like to approximate its (unknown) ground state. A common approach to this problem is variational ground state preparation using the now-standard variational quantum eigensolver (VQE) method~\cite{peruzzo2014variational,bauer2016hybrid,kandala2017hardware,tilly2022variational}, which readily generalizes to holographic qMPS simulations (holoVQE)~\cite{foss2021holographic}. 
Here, one variationally minimizes the expected energy:
\begin{align}
E_{\rm VQE}(\vec{\varphi}) =\<\psi_{\rm MPS}(\vec{\varphi})|H|\psi_{\rm MPS}(\vec{\varphi})\>,
\end{align}
 with respect to the variational parameters, $\vec{\varphi}$, where $H$ is the Hamiltonian for the spin chain in question, and $|\psi_{\rm MPS}\>$ is a qMPS. 
 Empirical studies of various physical models reveal that the ground-state energy tends to decrease polynomially in the number of variational parameters~\cite{haghshenas2021variational}.
Since we have seen that pulse level control can significantly outperform gate-based methods for unitary synthesis, we adopt an analog approach to VQE in which discretized versions of the drive waveforms are treated as variational parameters. 
This analog (``gate-free") approach to VQE, dubbed ctrl-VQE was introduced in~\cite{meitei2021gate} and modeled in qubit-only architectures. 
Here we numerically explore combining ctrl-VQE and holoVQE approaches in a coupled transmon/cavity device.
In contrast to unitary synthesis, VQE can be carried out without prior classical knowledge of the ground state.
In principle, for deep circuits, VQE approaches can face a complex, high-dimensional optimization landscape with local traps and barren plateaus. To mitigate these issues, one could follow a parallel/iterative optimization method~\cite{haghshenas2021variational} that incrementally grows the number of variational parameters (see also Appendix~\ref{supp:opt}), which have empirically proven successful for qMPS applications. However, for our purposes, we find that it suffices to adopt a simpler approach, in which we select the best result from a batch of 50 randomly initialized runs.

A practical issue is that classically simulating a cQED device requires truncating the infinite-dimensional cavity space to a finite number of levels. While many possible truncation schemes exist, since high-$n$ states of the cavity will have a shorter coherence time, we choose to truncate by simulating only cavity occupation numbers below a cutoff $\Lambda$. To ensure the optimizer does not take advantage of this artificial cutoff, we add a penalty term to $E_{\rm VQE}$ that penalizes occupation numbers in the range $\Lambda' \leq n \leq \Lambda$, where $\Lambda'$ is a fixed amount lower than $\Lambda$ (in practice we choose $\Lambda'=\Lambda-6$.
%We then set a target bond-dimension $D$, add a penalty term to the objective function $E_{\rm VQE}$, that and penalizes the occupation of states with $D< n \leq \Lambda$ (see Appendix~\ref{} for details), to suppress occupation of high-$n$ states, and ensure that the results do not depend sensitively on the cutoff $\Lambda$.

Fig~\ref{fig:result}a shows the resulting energy error $1-E_{\rm VQE}/E_0$ for the critical Ising model in Eq.~\ref{eq:sdim}, for various $\Lambda$, and $\tau$ from the holoVQE approach. Here, we estimate the true ground state energy, $E_0$, from classical DMRG with large bond dimension DMRG with $D=80$. 
%\acp{I recommend we remove the unitary-synthesis data from this figure, and perhaps move it to a panel in the unitary-synthesis figure above since it's a bit confusing to discuss all together here}
We observe that the holoVQE energy errors initially decrease with the simulation-size cutoff, $\Lambda$, before converging to a fixed value that is insensitive to the artificial cutoff. This convergence reflects that the simulation value should match the true behavior of the device.
Fig~\ref{fig:result}b shows the correlation functions, $\<\psi|\sigma^z_r\sigma^z_0|\psi\>$% \acp{let's switch $S = \sigma/2 \rightarrow \sigma$, and adjust the plot y-axis by a factor of $4$} 
at different spatial separations, $r$ for the converged, optimized qMPS for different $\tau$. We observe that the correlations are accurately captured out to several sites, and display the qualitative oscillations at much longer ranges, and that increasing $\tau$ gradually reduces the discrepancy between the qMPS and large-$D$ DMRG (essentially exact) correlations.

\paragraph{Impact of decoherence}
In practice, the continued improvement of accuracy with the length of the drive waveform will ultimately be limited by decoherence effects, which have so far been neglected in our simulations.
To assess the practical limitations due to noise and decoherence, we incorporate amplitude decay and dephasing processes into the holoVQE simulations (see Appendix~\ref{supp:noi} for details), and re-optimize the results in the presence of noise. The results are shown for various wave-form time $T$ and decay ($T_1$) and dephasing ($T_2$) times for the qubit and cavity. One observes a clear minimum in the relative energy error as a function of $T$ reflecting a tradeoff between the increase in control and the increased impact of noise with increasing length of drive wave-form. 

For the $T_{1,2}$ times reported in~\cite{heeres2017implementing}, the holoVQE approach can effectively utilize $D\approx 6$ levels of the cavity. By comparison, accessing qMPS with this bond dimension in a qubit-only architecture would require $1+\lceil \log_2D\rceil \approx 4$ qubits, whereas, in the cQED setting, it can be achieved with only a single device, showing the potential advantage of the latter. 
Improved coherence times, would amplify this hardware advantage by allowing control over an even larger number of quantum levels per device.

\paragraph{Discussion}
This numerical study demonstrates the viability of using cQED architectures for sequential/holographic simulation of many-body systems. The larger Hilbert space and longer coherence time of the cavity modes enable one to access larger bond-dimension (more entangled) quantum many-body states per device. 
The longer coherence time of the cavity mode also offers an advantage in this setting, since the less coherent qubit is frequently measured and reset, limiting the build-up of errors.
Based on these results, we estimate that, with the present technology, a single cQED device could be used to simulate an entangled many-body spin chain with several-site correlations.

There are a number of natural avenues to build on and extend this work. 
The size and complexity of the sequential/holographic simulations could be dramatically extended by introducing multiple cavity modes~\cite{naik2017random,chakram2021seamless}.
Assuming that $D\approx 6-10$ levels of each cavity mode can be controlled (which our simulations indicate is realistic for near-term accessible noise rates), then an $N$-cavity system would give access to bond dimension $D^N$, which would naively enable access to classically inaccessible bond-dimensions for $N$ as low as $5-10$.
The parallel control of multiple cavities using a single qubit has already been successfully demonstrated. A key question for future theoretical investigation is how well such multi-cavity devices can be controlled to perform holographic simulations of complex many-body states relevant to material science or chemistry.
Another possibility would be to encode a logical qubit into the cavity~\cite{heeres2017implementing} in order to integrate error-correction/suppression schemes~\cite{gottesman2001encoding,mirrahimi2014dynamically,ofek2016extending} directly into holographic simulations, with much less hardware overhead than in qubit-only architectures.
Finally, there has been significant recent progress in using cQED devices to engineer tailored forms of dissipation to synthesize interesting quantum channels~\cite{wang2019autonomous,cattaneo2021engineering,chen2022tuning}. It may be interesting to adapt these methods to implement quantum channels that emulate the transfer matrix of a qMPS representation for physical states.

\vspace{4pt}\noindent{\it Acknowledgements -- } 
We thank Michael Foss-Feig and Michael Zaletel for their insightful conversations. This work was supported by the US Department of Energy DOE DE-SC0022102.  A.C.P. thanks the Aspen Center for Physics where part of this work was completed. Y.Z. extends gratitude to Timothy Hsieh and the Perimeter Institute for Theoretical Physics for their hospitality, where a portion of this work was finished.

\bibliography{cqedbib.bib}
\appendix{}
\section{Full control with SNAP circuits}\label{supp:snap}
Previous results~\cite{lloyd1999quantum,krastanov2015universal} have shown that the universal control over the cavity can be accomplished with displacement operations and gates with the form: $e^{i \theta(n)\otimes\openone}$, which is the action of the SNAP gate with transmon set to $|0\>$ throughout the gate operation. By contrast, qMPS applications require preparing the transmon in an arbitrary superposition of $|0\>$ and $|1\>$ that are entangled with the cavity mode. 
Hence the SNAP gate is actually a controlled-SNAP gate: $e^{i \theta(n)\otimes\sigma^z}$. In this section, we show that combining this transmon-controlled SNAP gate with cavity displacements and arbitrary qubit rotations results in a universal gate set.
\newtheorem{theorem}{Theorem}
\begin{theorem}
The gate set: $\{D(\alpha), S(\theta(\circ)), R(\vec{\phi})\}$ is universal.
\end{theorem}
\begin{proof}
To establish the universality, it suffices to confirm that one can generate any arbitrary polynomial of the form $q^rp^s,q^jp^k\vec{\sigma}$, which forms a complete basis for the Lie algebra of the qubit and oscillator. Here $q = \frac{1}{\sqrt{2}}(a+a^\dagger)$ and $p= \frac{i}{\sqrt{2}}(a^\dagger-a)$ are the (dimensionless) oscillator coordinate and momentum, $r,s$ are arbitrary positive integers.
Fixing $\theta(n) = \eps n$, and considering infinitesimal rotations, and displacements, the generators for the gate set $\{D(\alpha), S(\theta(\circ)), R(\vec{\phi})\}$ include: $\{n\sigma^{x,y,z},q, p,\sigma^{x,y}\}$. 
By repeated use of the BCH formula $e^{i\eps A}e^{i\eps B}e^{-i\eps A}e^{-i\eps B} \approx e^{-\eps^2[A,B]}+\mathcal{O}(\eps^3)$, one can generate any commutator of pairs of generators.
Then, combining the following sequence of commutator identities
\begin{align}
\[n\sigma^z,q\]&\sim p\sigma^z \nonumber\\
\[q\sigma^x,q\sigma^y\] &\sim q^2 \nonumber \sigma^z\\
\[q^{r+1}p^s\sigma^z,p\sigma^z\] &\sim q^rp^s \nonumber\\
\[q^{r+1}p^s\sigma^z,p\] &\sim q^rp^s\sigma^z,
\end{align}
establishes, through nested applications of the above BCH formula, that one can indeed generate the polynomials required for universal control with this gate set.
\end{proof}
To empirically verify the universal control for this circuit, we conduct a small batch of simulations (Fig.~\ref{fig:snap_proof}) using optimal control techniques to synthesize Haar random unitaries over an eight-level Hilbert space with SNAP/displacement circuits.

\begin{figure}
    \centering.
    \includegraphics[width=.42\textwidth]{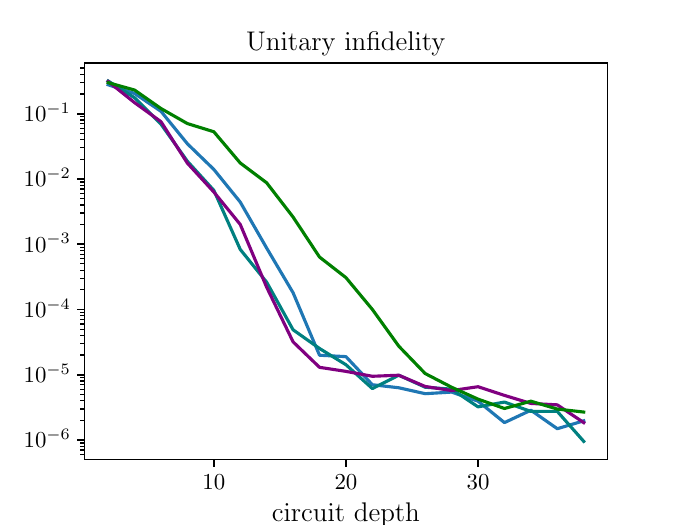}
    \caption{{\bf Synthesizing Haar random unitaries}. We synthesize $8\times8$ Haar random unitaries over the cavity and the transmon with the SNAP circuit demonstrated in Fig.~\ref{fig:schematic}, and we observe a similar result to the study in~\cite{fosel2020efficient}, where a universal control over the cavity alone is aimed. Each curve stands for a different random instance.
}
    \label{fig:snap_proof}
\end{figure}

\paragraph{SNAP Implementation time} 
To directly compare the time to synthesize quantum operations to implement a known MPS with either SNAP or GRAPE methods, we need to estimate the time required to perform each layer of the SNAP-displacement circuit.
In this appendix, we estimate this time using the results of previous optimal-control studies of synthesizing SNAP gates~\cite{krastanov2015universal}.
 The first proposed implementations are to implement a SNAP gate with a pair of $\pi$ pulses with carrier frequency, $\approx \omega_c + n\chi$ associated with a particular cavity mode occupation, $n$. The associated phase is controlled by offsetting the phase of the two pulses~\cite{krastanov2015universal}. Multiple such cavity-dependent rotations can be implemented in parallel by frequency multiplexing.
Assuming the absence of higher order dispersive shifts, using this method with Gaussian pulses requires time $\sim \(2\pi/\chi\)\sqrt{\log1/\epsilon}$ to implement a SNAP gate with infidelity $\sim \epsilon$ due to off-resonant coupling. For, example with the $\chi$ used in our numerics, achieving a $\epsilon \lesssim 10^{-4}$ SNAP would require time $\sim 1000$ ns. 
 This two-pulse approach can be improved using optimal control theory (OCT) methods to variationally optimize a SNAP waveform~\cite{kudra2022robust}. 
 %
%To improve gate fidelity, Kudra et al. proposed a variational method which reduced the implementation time cost by an order of magnitude~\cite{kudra2022robust}. 
With a qubit-cavity dispersive shift at $2\pi\times3.14$ MHz, which is 1.4 times faster than our $\chi$, Ref.~\cite{kudra2022robust} reached a 500 ns implementation time with high fidelity. It is natural to assume that the implementation time is inversely proportional to $\chi$, in which case this translates to $\sim700$ ns for the parameters used in our simulations. On the other hand, displacement operations can be implemented by resonantly driving the cavity with pulses with a sine-squared envelope and a calibrated amplitude proportional to $\alpha$, which is much faster than the time required by a SNAP gate and can be implemented within less than 100 ns~\cite{axline2018demand,fosel2020efficient}.  Therefore, we estimate the total time for implementing each layer of the SNAP circuit as $\approx800$ ns. In our simulations, we treat each SNAP gate as perfectly implemented.

\section{Optimization details}\label{supp:opt}
\paragraph{Isometries to target unitaries}As mentioned above, with both methods, we would like to compare the time cost of a classically calculated $U_{\text{DMRG}}$ to some certain accuracy. It's worth noticing that in cQED simulation, we are truncating the experimentally infinite level Hilbert space using a finite cut-off $m_c$. Suppose setting $m_c = m_b$, the one has no information at all about what the experimentally implemented unitary does in an infinite Hilbert space: the wave-function will eventually occupy higher levels. Therefore, we embed the target unitary into a Hilbert space with higher cavity levels:
\begin{align}U_{\text{targ}} =
\begin{pmatrix} 
U_{\text{DMRG}} & 0 \\
0 & \mathbbm{1}
\end{pmatrix}
\end{align}
That is, we demand the total time evolution returns identity on the cavity levels between $m_b$ and $m_c$, forming a ``buffer'' between the finite logical space and the rest of the (unknown) infinite Hilbert space. We numerically certify that $m_c\lesssim 2 m_b$ suffices to for the time length we want.
\paragraph{Batch optimization} Optimizing a large parametrized quantum circuit (PQC) can be difficult due to the so-called barren plateau problems. Both random initialization and gradient-free methods can result in the gradient of the objective function becoming negligible. To combat this, we've adopted a batch-sequential optimization strategy that is first proposed in~\cite{haghshenas2021variational} and developed in~\cite{zhang2022qubit}.

For the SNAP circuit optimization, we start with a batch of single-layer circuits, each with randomly initialized parameters. We then optimize each circuit in the batch using a local optimizer, selecting the circuit with the best outcome. This selected circuit is used to generate another batch of circuits with an additional identity gate layer and slight randomness in gate parameters. This approach retains the desirable features of the optimized first layer while also offering the opportunity to escape from local minima in the target function landscape.

We repeat this process with each new batch, increasing the depth of the circuits until we reach our target. As the number of parameters increases, we decrease the amount of randomness to ensure a good performance.\\

\section{Simulating noise}\label{supp:noi}
To study the effect of decoherence sources and improve optimization results, we simulate the system dynamics using a discretized master equation
with the known physical parameters. We start with the continuous-time Lindblad master equation:
%\begin{align}
%\frac{\d \rho}{\d t} =  &-i[H(t),\rho(t)]\\
%                                &+\frac{1}{T_{1,c}}\mathcal{D}[a]\rho +\frac{1}{T_{1,q}}\mathcal{D}[b]\rho +\frac{1}{T_{2,q}}\mathcal{D}[b^{\dag}b]\rho 
%\end{align}
\begin{align}
\frac{\d \rho}{\d t} =  &-i[H(t),\rho(t)]\\
                                &+\frac{1}{T_{1,c}}\mathcal{D}[a]\rho +\frac{1}{T_{1,q}}\mathcal{D}[\sigma^-]\rho +\frac{1}{4T_{2,q}}\mathcal{D}[\sigma^z]\rho 
\end{align}
where $T_{1,c} = 2700\ \mu s $, $T_{1,q} = 170\ \mu s$ are the relaxation times for the oscillator and transmon; $T_{2,q} = 43\ \mu s$ is the transmon dephasing time; $a$ and $b$ stand for the annihilation operators for the oscillator and transmon, respectively. Here the ``dissipator" $\mathcal{D}$ is defined as:
\begin{align}
\mathcal{D}[\mathcal{O}]\rho = \mathcal{O}\rho \mathcal{O}^{\dag} - \frac{1}{2} \{\mathcal{O}^{\dag} \mathcal{O}, \rho\}
\end{align}
The Lindblad equation can be thought of continuous measurement performed on a quantum dynamical system and it forms a complete positive trace-preserving (CPTP) mapping. To capture the noise, the tensor-network formalism over some short period of time, $\tau$, we cast it as a discretized quantum channel: 
\begin{align}
\mathcal{L}:\rho\rightarrow \mathcal E(U^\dag\rho U)
\end{align}
where 
%$$ \mathcal E(\rho) = (\openone + \frac{\tau}{T_{1,c}}\mathcal{D}[a] +\frac{\tau}{T_{1,q}}\mathcal{D}[b] +\frac{\tau}{T_{2,q}}\mathcal{D}[b^{\dag}b])\rho $$
$$ \mathcal E(\rho) = (\openone + \frac{\tau}{T_{1,c}}\mathcal{D}[a] +\frac{\tau}{T_{1,q}}\mathcal{D}[\sigma^-] +\frac{\tau}{4T_{2,q}}\mathcal{D}[\sigma^z])\rho $$
which is illustrated in Fig.~\ref{fig:noise}.
%\newpage\hrule
%\acp{Need tensor network diagram showing how this is implemented}

\begin{figure}[b!]
    \centering.
    \includegraphics[width=.42\textwidth]{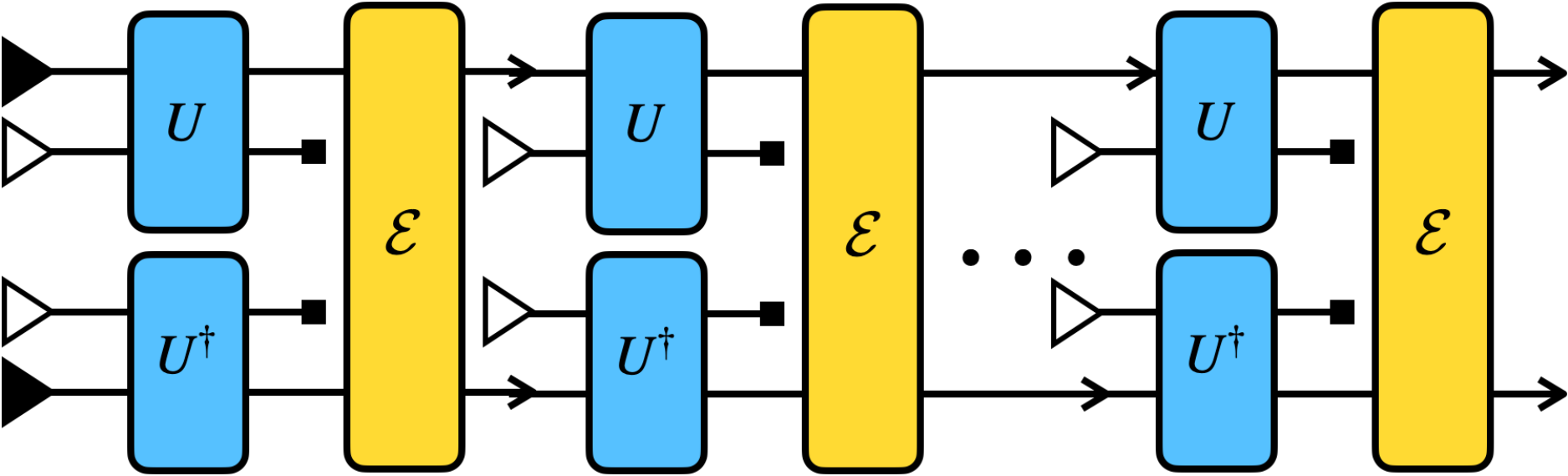}
    \caption{{\bf Simulating noise}. A matrix product density operator representation of our discretized Lindbladian simulation. The black and white triangles stand for zero states on the cavity and the transmon.}
    \label{fig:noise}
\end{figure}

{}

\end{document}